\documentclass[pra,english,reprint,longbibliography, superscriptaddress, breaklinks=true, showkeys, showpacs=false, nofootinbib]{revtex4-2}

\usepackage[T1]{fontenc}
\usepackage[utf8]{inputenc}
\usepackage{color}
\usepackage{babel}
\usepackage{amsmath}
\usepackage{amsthm}
\usepackage{amssymb}
\usepackage{subfigure}
\usepackage{graphicx}
\usepackage{physics} 
\usepackage{dsfont}
\usepackage{hyperref}
\usepackage{float}
\usepackage{tabularray}
\usepackage{tabularx}
\usepackage{array}
\usepackage{subcaption}
\usepackage{ragged2e}

\newtheorem{theorem}{Theorem}
\usepackage[dvipsnames]{xcolor}
\usepackage{mathrsfs}

\begin{document}

\title{Geometric quantum thermodynamics: A fiber bundle approach}

\author{Tiago Pernambuco}
\email{tiago4iece@gmail.com}
\affiliation{Theoretical and Experimental Physics Department, Federal University of Rio Grande do Norte, 59078-970, Natal, Brazil}

\author{Lucas C. C\'eleri}
\email{lucas@qpequi.com}
\affiliation{QPequi Group, Institute of Physics, Federal University of Goi\'as, Goi\^ania, Goi\'as, 74.690-900, Brazil}

\begin{abstract}
Classical thermodynamics is a theory based on coarse-graining, meaning that the thermodynamic variables arise from discarding information related to the microscopic features of the system at hand. In quantum mechanics, however, where one has a high degree of control over microscopic systems, information theory plays an important role in describing the thermal properties of quantum systems. Recently, a new approach has been proposed in the form of quantum thermodynamic gauge theory, where the notion of redundant information arises from a group of physically motivated gauge transformations called the thermodynamic group. In this work, we explore the geometrical structure of quantum thermodynamics. In particular, we do so by explicitly constructing the relevant principal fiber bundle. We then show that there are two distinct (albeit related) geometric structures associated with the gauge theory of quantum thermodynamics. In this way, we express thermodynamics in the same mathematical (geometric) language as the fundamental theories of physics. Finally, we discuss how the geometric and topological properties of these structures may help explain fundamental properties of thermodynamics.
\end{abstract}

\maketitle

\section{Introduction}

The universal character of thermodynamics stems from the type of physical questions it is designed to address. In the thermodynamic limit, experimentally accessible observables probe only coarse-grained averages of the underlying microscopic dynamics, since atomic-scale details lie beyond direct measurement. This coarse-graining gives rise to thermodynamic state variables and results in a robust theoretical framework that constrains macroscopic processes largely independently of microscopic details~\cite{Callen1991}. An essential consequence of this limit is the effective suppression of fluctuations. For small-scale or quantum systems, by contrast, fluctuations play a central role, and under far-from-equilibrium driving, the standard thermodynamic description breaks down. In such regimes, alternative approaches must be employed, such as stochastic thermodynamics~\cite{Strasberg2022} and information theory~\cite{Goold_2016}.

Despite significant advances since Alicki’s pioneering work~\cite{Alicki1979}, quantum thermodynamics remains incomplete~\cite{Alicki2018}. Foundational questions persist with respect to the definitions of work, heat, and entropy in the quantum regime; see~\cite{Deffner2019,Strasberg2022} for recent discussions. 

In contrast to thermodynamics, which emerges from coarse-grained descriptions of high-dimensional systems, modern fundamental physics---the standard model and general relativity---is formulated in terms of gauge theories~\cite{Baez1994,konopleva1981gauge,faddeev1991gauge}. A gauge theory is defined by a Lagrangian that is invariant under local transformations acting on internal degrees of freedom; the requirement of local invariance demands the introduction of gauge fields that mediate interactions.

Geometrically, gauge theories are expressed in the language of fiber bundles: gauge fields correspond to connections on principal bundles, and their field strengths correspond to the associated curvatures. The base manifold represents spacetime, and the fibers encode the internal symmetry group~\cite{Baez1994,konopleva1981gauge,faddeev1991gauge,nash2013topology,bleecker2013gauge,konopleva1981gauge}. Although gauge theories originated in particle physics, they appear in various contexts, including low Reynolds number hydrodynamics~\cite{PhysRevLett.58.2051,Shapere_Wilczek_1989}, condensed matter physics~\cite{Fradkin2013}, quantum information~\cite{GNNGauge2021}, and finance~\cite{Vazquez2012}, to mention just a few.

Classical thermodynamics has also been explored from this gauge-theoretic perspective. Treating the thermodynamic force as a gauge field, Ref.~\cite{Katagiri2018} analyzed the Onsager relations within nonequilibrium thermodynamics, while Ref.~\cite{Borlenghi2016} revealed a connection between $U(1)$ lattice gauge theories and nonequilibrium processes. Moreover, by examining the symmetry under local rescaling of probabilities in Markovian dynamics on finite graphs, Ref.~\cite{polettini2012nonequilibrium} proposed interpreting thermodynamic forces as gauge potentials.

Based on information theory, Refs.~\cite{ThermoGauge1,ThermoGauge2} cast quantum thermodynamics as a gauge theory of the thermodynamic group. Classical thermodynamics relies on coarse-graining, whereas quantum mechanics assumes full control over microscopic degrees of freedom, yielding an excess of information from a thermodynamic standpoint. The thermodynamic gauge group removes this redundant information, playing a role analogous to the elimination of microscopically irrelevant details in classical thermodynamics. Within this formalism, thermodynamic quantities arise uniquely from gauge invariance, which is one of the most fundamental principles of modern physics.

The gauge structure introduced in~\cite{ThermoGauge1,ThermoGauge2} is particularly rich, featuring a time-dependent gauge group whose geometric properties remain unexplored. The goal of this work is to clarify the underlying geometry of this theory and to construct the associated fiber-bundle structures in a mathematically rigorous way. This puts thermodynamics on the same mathematical structure as fundamental physical theories. This may shed new light on the foundations of classical and quantum thermodynamics. 

The paper is structured as follows. In Sec.~\ref{ThermoGaugeTheory}, we review the gauge theory of the thermodynamics group introduced in Refs.~\cite{ThermoGauge1,ThermoGauge2}. Section~\ref{GaugeGeometry} presents how gauge theories can be formulated in terms of the geometry of principal fiber bundles and introduces some common examples from field theory to pave the way for the construction of our theory, which is done in Sec.~\ref{Geometrization}. Finally, in Sec.~\ref{Discussion}, we discuss how our geometric theory may help us understand thermodynamics and provide new mathematical insights into the theory using geometric and topological techniques, as well as proposing future directions for investigation. Additionally, to make our discussion more self-contained, in Appendix~\ref{MathPrelim}, we revise the mathematics of fiber bundles, principal bundles, and Lie theory, which are used extensively throughout this paper. 

\section{Quantum thermodynamics as a gauge theory} \label{ThermoGaugeTheory}

Classical thermodynamics is a theory concerned with the macroscopic bulk properties of large systems, where fluctuations vanish. The theory focuses on the average behavior of the microscopic degrees of freedom (justified by the central limit theorem), while ignoring the microscopic configurations. This means that, in classical thermodynamics, we have limited access to information, which is called coarse graining~\cite{Callen1991}.

In trying to adapt useful concepts of thermodynamics to the quantum realm, quantum thermodynamics has emerged as a powerful theoretical framework~\cite{Deffner2019,Strasberg2022}. In great contrast to its classical counterpart, however, quantum thermodynamics is usually based on a significant amount of information about the system being studied. For instance, thermodynamic protocols are frequently considered in which knowledge of the density operator is mandatory. This is why quantum thermodynamic quantities are often defined in terms of informational measures~\cite{Deffner2019,Goold_2016}.

In an attempt to formulate quantum thermodynamics under the same paradigm that underlying classical thermodynamics, which encompasses the coarse-graining nature of classical macroscopic measurements, a new framework has been proposed~\cite{ThermoGauge1,ThermoGauge2}. The new theory proposes that when we limit our information about the system to that attainable by measurements of specific observables (thus not performing full state tomography), the full quantum state contains too much information, most of which is inaccessible to the observer. For example, if one can only perform energy measurements on the system, it is impossible to distinguish between any two states with the same energy or to obtain information regarding quantum coherences on this basis~\cite{ThermoGauge1,ThermoGauge2}.

\begin{figure}[t]
    \centering
    \includegraphics[width=1.0\linewidth]{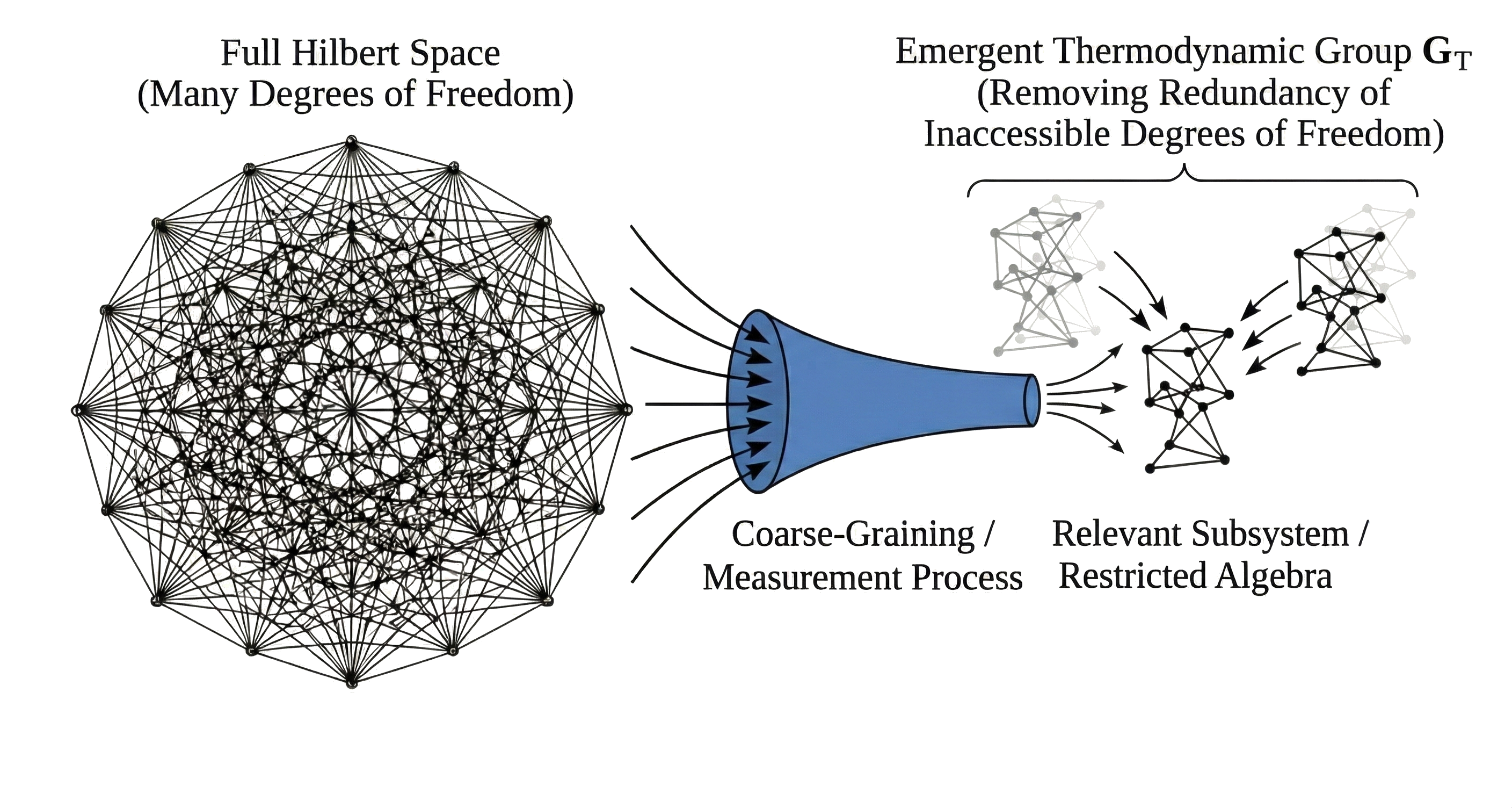}
    \caption{\justifying The coarse-graining process defined by restricting the measurements one may perform defines the emergent gauge group of quantum thermodynamics through the removal of redundancies in description due to inaccessible degrees of freedom.}
    \label{Coarse}
\end{figure}

To eliminate the redundancy created by this excess information, the theory introduces the concept of an emergent gauge symmetry. See Fig.~\ref{Coarse}. This symmetry is called emergent because the excess information contained in a quantum system is redundant only from the perspective of thermodynamics and not from information theory~\cite{ThermoGauge1}. This is analogous to the classical case, where the coarse graining is implemented by ignoring the microscopic degrees of freedom that are not accessible to classical measurements. Building on the notion of redundant information introduced in Refs.~\cite{ThermoGauge1,ThermoGauge2}, density matrices, which act as carriers of information, are proposed to be interpreted as analogues of gauge potentials in gauge theories.

Although we can consider any set of observables, let us consider the case where we can only perform energy measurements on a quantum system described by a time-dependent density matrix $\rho_t$. This is the typical situation on many experimental platforms when the dimension of the system increases~\cite{Araujo2018}. Then, for a general time-dependent Hamiltonian $H_t$, the mean energy of the system is
\begin{equation}
    U[\rho_t] = \Tr(\rho_t H_t).
\end{equation}
The main idea behind the theory is to identify a symmetry group such that any density operator linked to $\rho_t$ by the action of the group gives the same value for the mean energy. See Fig.~\ref{ThermoGroup2}. In this way, we have an equivalence class of systems such that
\begin{equation}
    U[\rho_t] = \Tr(V_t\rho_tV_t^\dagger H_t) = U[V_t\rho_tV_t^\dagger],
    \label{gaugeinv}
\end{equation}
with $V_t$ representing the symmetry transformation, an element of the thermodynamic group. Since the matrices $V_t$ take one density matrix to another, we can assume that they are unitary. Joining this fact with Eq.~\eqref{gaugeinv}, we see that the set of transformations $V_t$ allowed must commute with the Hamiltonian of the system. This leads to a gauge group that is isomorphic to the following Lie group~\cite{ThermoGauge2}:
\begin{equation}
\mathrm{G}_T = \mathrm{U}(n_t^1) \times \mathrm{U}(n_t^2) \times ... \times \mathrm{U}(n_t^k),
    \label{thermogroup}
\end{equation}
where $\times$ denotes the Cartesian product and $\mathrm{U}(n_t^i)$ is the unitary group of dimension $n_t^i$, the (in general, time-dependent) degeneracy of the $i$-th eigenvalue of the Hamiltonian. Thus, $\sum_i n_t^i = d$ is the dimension of the Hilbert space of the system. The group $\mathrm{G}_T$ is called the thermodynamic group. Note that the structure of the group will change if the set of observables is modified; however, the general construction remains the same. Gauge invariance is imposed here by stating that physical quantities, which are expressed as functions of the density matrix, must be invariant under the action of $\mathrm{G}_T$. 

\begin{figure}[t]
    \centering
    \includegraphics[width=1.0\linewidth]{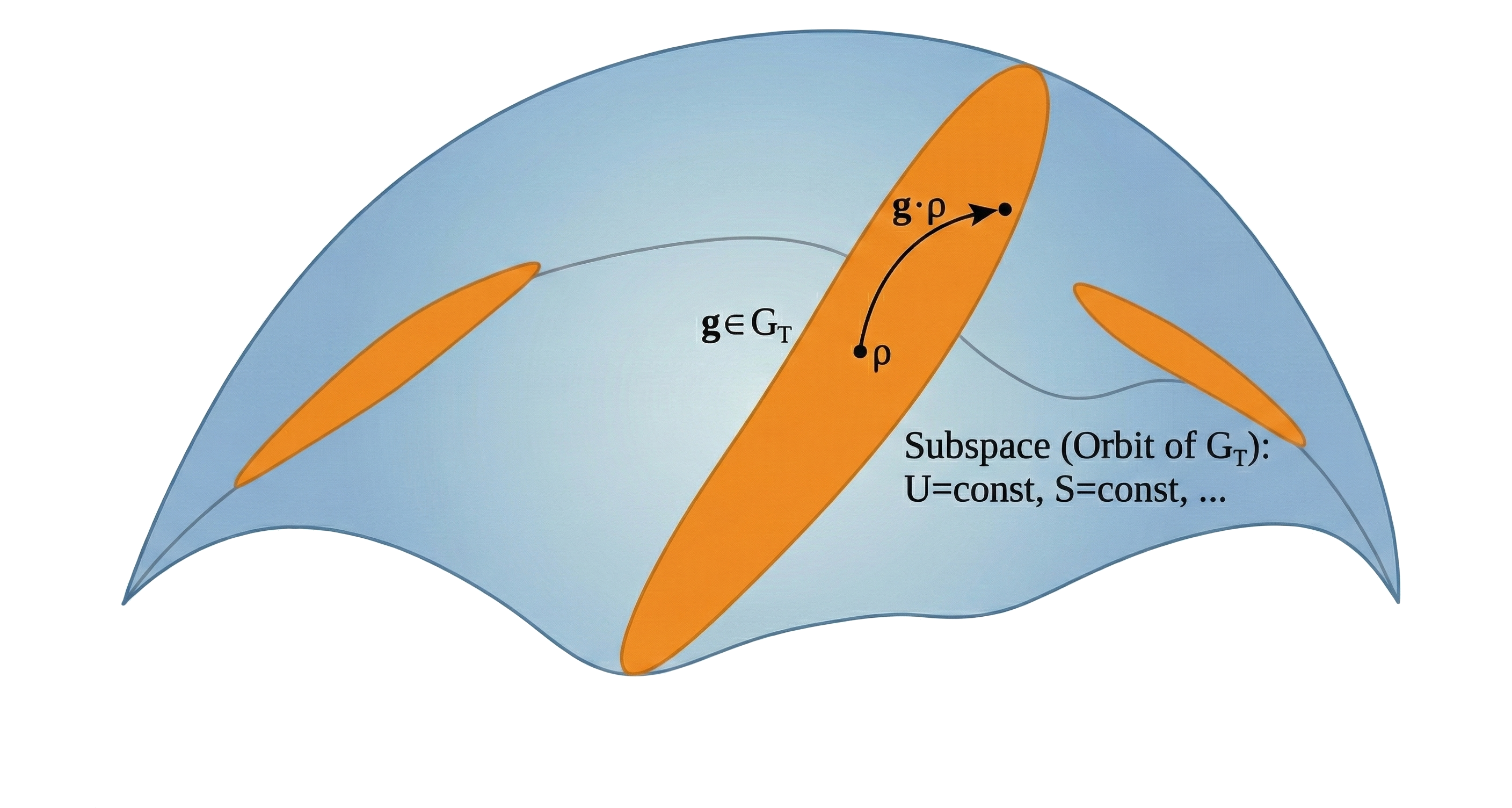}
    \caption{\justifying The action of the thermodynamic group $\mathrm{G}_T$. The orbits (orange blobs), created by the group action $g \cdot \rho$ (which is just conjugation by the matrices $V_t$) are subspaces of density matrices that are indistinguishable for the measurement scheme that defines the group.}
    \label{ThermoGroup2}
\end{figure}

In general, the degeneracy structure that defines $\mathrm{G}_T$ can vary over time, leading to a family of thermodynamic gauge groups $\mathrm{G}_T$. The physical interpretation and consequences of this time dependence are discussed in Sec.~\ref{Geometrization}.

In further analogy with modern gauge theories, Ref.~\cite{ThermoGauge1} introduces the notion of a covariant derivative $\mathrm{D}_t$ defined by a potential $A_t = i \dot u_t u_t^\dagger$, where $u_t$ is the unitary matrix that diagonalizes the chosen observable, the Hamiltonian of the system in this case. It is defined by the action
\begin{equation}
\mathrm{D}_t (\cdot) = \partial_t (\cdot) + i[A_t, (\cdot)],
\label{CovariantDerivative}
\end{equation}
Based on this covariant derivative, one can write the gauge-invariant definitions of work $W_{inv}$ and heat $Q_{inv}$ for quantum systems~\cite{ThermoGauge1} as
\begin{equation}
W_{inv}[\rho_t] = \int_0^\tau \dd t \Tr(\rho_t \mathrm{D}_t H_t),
    \label{Winv}
\end{equation}
and
\begin{equation}
Q_{inv}[\rho_t] = \int_0^\tau \dd t \Tr(H_t \mathrm{D}_t \rho_t),
    \label{Qinv}
\end{equation}
respectively. In these equations, the integration is performed over time, considering a thermodynamic process that occurs between times $0$ and $\tau$. See Refs.~\cite{ThermoGauge1,ThermoGauge2} for the physical meaning of these quantities. It is important to note that we have two distinct actions on the system. First, we have a physical process under which the state evolves, and we have the action of the thermodynamic group, which is simply a symmetry action at each instant of time. Work and heat, as well as other physical quantities, are related to the physical process.

Another important thermodynamic functional that can be defined in a gauge-invariant fashion is the entropy, which takes the form~\cite{ThermoGauge2}
\begin{equation}
    S_{\mathrm G_T}[\rho_t] = S_{vn}[\rho^E_t],
    \label{SG}
\end{equation}
where $S_{vn}(\rho)=-\Tr\rho\ln\rho$ is the von Neumann entropy and 
\begin{equation}
    \rho^E_t = \bigoplus_{k=1}^{p \leq d}\frac{Tr[\Pi_{n_t^k}\rho_t\Pi_{n_t^k}]}{n_t^k} \mathbb I_{n_t^k}.
    \label{rhoddE}
\end{equation}
In Eq.~\eqref{rhoddE}, $\Pi_{n_t^k}$ are projectors onto the subspace associated with the degeneracy $n_t^k$, while $\mathbb I_{n_t^k}$ stands for the identity operator in that subspace. The index $E$ stands for the energy eigenbasis. If we choose a distinct observable, the structure of this operator will remain the same, but the eigenprojectors will be those of the other observable. See Ref.~\cite{ThermoGauge2} for more details.
    
Moreover, this procedure for constructing gauge-invariant quantities is not exclusive to entropy. In particular, since $\rho^E$ is the average density matrix with respect to the thermodynamic group (via quantum twirling), it makes sense to define a general gauge-invariant thermodynamic functional as~\cite{ThermoGauge2}
\begin{equation}
    F_{inv}[\rho_t] = F[\rho^E_t]
\end{equation}
as long as its non-gauge-invariant counterpart $F[\rho]$ is invariant under unitary transformations~\cite{ThermoGauge2}.

Thus, this theory defines physical quantities as those invariant under the action of the gauge group $\mathrm{G}_T$, in complete analogy to the gauge principle in fundamental theories of physics. Here, we take a step forward in the formalization of the theory by providing its constructions in terms of fiber bundles, the natural language in which modern gauge theories are formulated. 

\section{Gauge theories as geometry of principal bundles} \label{GaugeGeometry}

In modern mathematics and theoretical physics, fields are understood as sections of fiber bundles. This viewpoint makes explicit the geometric origin of local degrees of freedom, clarifies the role of gauge choices, and highlights how connections generalize ordinary directional derivatives. Fiber bundles, therefore, provide the natural mathematical framework for formulating such theories. For example, in gauge theories describing the fundamental interactions between elementary particles, the interaction is mediated by a gauge potential. Geometrically, a gauge potential is the local expression of a connection that lives on a principal fiber bundle over the four-dimensional spacetime.

The purpose of this section is twofold: first, to establish the notation used in the next section; and second, to give a concise overview of the main mathematical structures involved. Readers seeking a deeper treatment are referred to Refs.~\cite{Baez1994,konopleva1981gauge}, while Appendix~\ref{MathPrelim} provides additional details.
 
Given two topological spaces $\mathds{E}$ (referred to as the total space) and $\mathds{B}$ (referred to as the base space), a Lie group $\mathrm{G}$ (the fiber), and a projection map $\pi~:~\mathds{E}~\rightarrow~\mathds{B}$, a principal $\mathrm{G}$-bundle $\xi=(\mathds{E},\pi,\mathrm{G},\mathds{B})$ is a geometric structure such that there exists an open covering $\{O_\alpha\}$ and a set of homeomorphisms $\phi_\alpha:\pi^{-1}(O_\alpha)\rightarrow O_\alpha \times \mathrm{G}$ (called local trivializations of $\mathds{E}$) satisfying $\pi\phi_\alpha^{-1}(x, g) = x$, for $x \in O_\alpha$ and $g \in \mathrm{G}$~\cite{nash2013topology}. Figure~\ref{bundle} presents a schematic representation of these concepts.

\begin{figure}
    \includegraphics[width=0.5\textwidth]{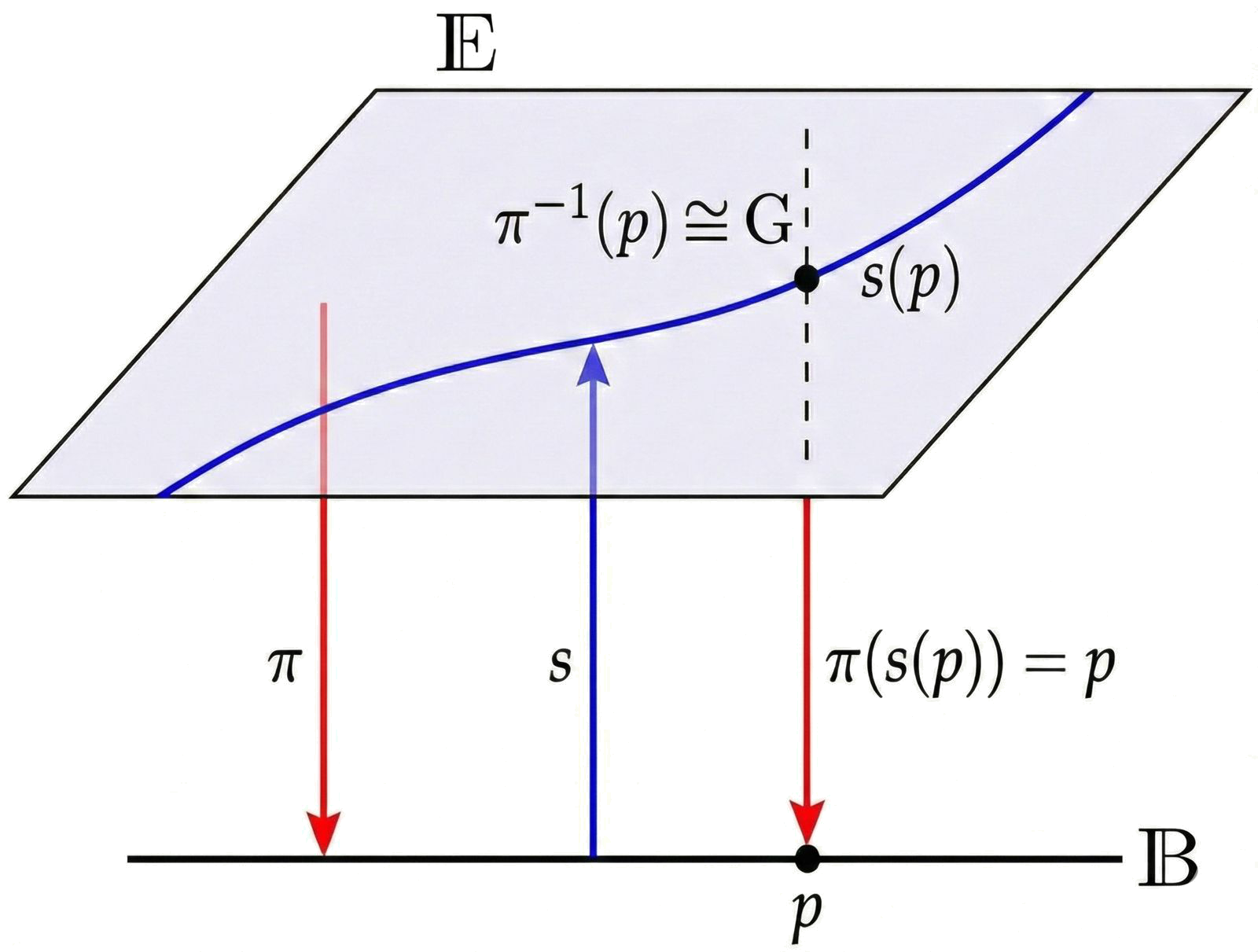}
    \caption{\justifying A simple example of a bundle. The base space $\mathds{B}$ is the real line $\mathds{R}$, while the total space $\mathds{E}$ is the plane $\mathds{R}\times\mathds{R}$. The projection map $\pi$ associates a fiber $E_{p}$ to each point of $\mathds{B}$. The total space is the union of all the fibers $\mathds{E} = \cup_{p\in\mathds{E}}{E}_{p}$. The projection simply tells you which point of the base space a given fiber element belongs to. In physics, the base space is usually spacetime and the fiber encodes internal degrees of freedom attached to each spacetime point. The section $s$ of the bundle is also shown. A section is a choice of exactly one fiber element over each point of the base space. In physics, a section is the rule that assigns to each point in spacetime the value of a physical quantity living in the fiber.}
    \label{bundle}
\end{figure}

The above definition reflects the idea that the total space of a principal bundle is locally, but not necessarily globally, a product space of $\mathds{B}$ and $\mathrm{G}$. In the case where $\mathds{E}$ is globally isomorphic to $\mathds{B} \times \mathrm{G}$, we say it is a trivial principal bundle~\cite{nash2013topology}. We can also define a function $s~:~\mathds{B}~\rightarrow~\mathds{E}$ called a section of the bundle that locally inverts the projection $\pi$. See Fig.~\ref{bundle} for more details of this definition. A principal bundle is trivial if and only if there exists a global section; i.e., the projection can be globally inverted. The sections of the principal fiber bundles in the context of gauge theories are associated with the choices of gauge~\cite{nash2013topology,bleecker2013gauge}.

In the geometric theory of gauge fields, one is generally interested in the following setup~\cite{Baez1994,bleecker2013gauge, konopleva1981gauge,nash2013topology}: First, in the principal $\mathrm{G}$-bundle $\xi$, the gauge group of the theory is $\mathrm{G}$ (e.g. $\mathrm{U}(1)$ for electromagnetism, $\mathrm{SU}(2)$ for the weak interaction, and $\mathrm{SU}(3)$ for the strong interaction). Most often, $\xi$ is the bundle product with $\mathds{B}=\mathcal{M}$, where $\mathcal{M}$ is the Minkowski spacetime of dimensions $(3 + 1)$. This is because $\mathcal{M}$ is effectively $\mathbb{R}^4$ with a certain choice of metric; therefore, it is contractible~\cite{nash2013topology}. Since $\mathcal{M}$ is contractible, according to Theorem~\ref{contract} in Appendix~\ref{MathPrelim}, it is always trivial; therefore, $\xi$ is isomorphic to the product bundle, which is easier to work with in local coordinates.

The components of a connection $\omega$ in $\xi$ are identified with the gauge fields of the theory, while the curvature form $\Omega$ associated with $\omega$ is interpreted as the field strength, and the gauge transformations $\omega \rightarrow g^{-1}\omega g + g^{-1}\dd g$ of the potentials are simply the transformation law of the connection form for $g \in \mathrm{G}$ and $\mathrm{G}$, a matrix Lie group~\cite{bleecker2013gauge}.

Given $\xi$, we also have a vector bundle (a fiber bundle in which the fiber has a vector space structure~\cite{nash2013topology,tu2017differential}) called the associated bundle. The sections of this bundle are the particles of the theory~\cite{bleecker2013gauge}.

The associated bundle of a principal bundle can be constructed in the following manner~\cite{tu2017differential}: given $\xi$ and a linear representation $\rho$ of $\mathrm{G}$ acting in a vector space $\mathds{V}$, the associated bundle $\eta_\rho$ is the quotient of $\mathds{E} \times \mathds{V}$ by the following equivalence relation \cite{tu2017differential}:
\begin{equation}
(p, \nu) \sim (p\cdot g, \rho(g^{-1})\nu),   
\end{equation}
where $g \in \mathrm{G}, \ (p, \nu) \in (\mathds{E}\times\mathds{V})$. If $\xi$ is trivial, then $\eta_\rho$ is also trivial and isomorphic to the product vector bundle with base space $\mathds{B}$ and fiber $\mathds{V}$~\cite{tu2017differential}. Furthermore, given the associated bundle $\eta_\rho$, any connection $\omega$ in $\xi$ induces a covariant derivative $\mathrm{D} = \dd + \omega$ in tensorial forms valued $\mathds{V}$ in $\mathds{E}$. This is the familiar covariant derivative from gauge theory~\cite{Baez1994,konopleva1981gauge, bleecker2013gauge,faddeev1991gauge}.

In order to make the ideas clearer, let us consider the electromagnetic case, which is the simplest gauge theory we have in fundamental physics. Electrodynamics can be constructed geometrically on a principal $\mathds{U}(1)$-bundle over Minkowski spacetime $\mathcal{M}$~\cite{nash2013topology, bleecker2013gauge, konopleva1981gauge}. In this construction, since the Lie algebra $\mathfrak{u}(1)$ is just the imaginary numbers, we can associate the components $A_\mu$ of the vector potential with the coefficients of a 1-form connection in this fiber bundle by $\omega_u = -iA_\mu$~\cite{bleecker2013gauge}. The connection form itself is then written as $\omega = \omega_\mu \dd x^\mu$. 

As mentioned above, for a principal fiber bundle with a matrix Lie group as fiber, the connection transforms as $\omega \rightarrow g^{-1}\omega g + g^{-1} \dd g$. Since $\mathrm{U}(1)$ is abelian and $ \mathrm{U}(1) = \{e^{i\phi}, \theta \in \mathds{R}\}$, this reduces to $\omega \rightarrow \omega + \dd\phi$. Noting that $\phi$ is a 0-form, one immediately sees that this implies the usual gauge transformation of electromagnetism.

From the connection $\omega$, one can compute a curvature 2-form using Eq.~\eqref{CurvatureForm} of Appendix~\ref{MathPrelim}. The Lie bracket term disappears because $\mathrm{U}(1)$ is abelian, so $\Omega = dd\omega$. In terms of coefficients, this becomes $\Omega_{\mu \nu} = \partial_\mu \omega_\nu - \partial_\nu \omega_\mu$. In terms of the gauge potential, which is the connection, this is translated to $F = \dd A$, where $F$ is the electromagnetic field tensor~\cite{Baez1994,nash2013topology}.

In a way, the entirety of classical electrodynamics is contained within the tensor $F$. The four Maxwell equations (in the absence of sources) can be elegantly written as the Yang-Mills equations
\begin{align}
    &\dd F = 0,\nonumber \\
    &\dd \star F = 0,
    \label{Maxwell}
\end{align}
where $\star$ is the Hodge star operator. By analyzing them component-wise, it can be directly seen that Eqs.~\eqref{Maxwell} lead to all four Maxwell's equations~\cite{Baez1994,nash2013topology,konopleva1981gauge}. This kind of expression is advantageous because it is valid for an arbitrary base manifold and does not depend on any specific set of coordinates~\cite{konopleva1981gauge}.

Lastly, it would be of interest to take a look at how charged particle fields arise in this formalism. Following the construction we did earlier for the associated bundle, we must choose a representation of the group $\mathrm{U}(1)$. All linear representations of the unitary 1-dimensional group are of the form $\rho(g) = e^{iqt}$, $q \in \mathbb Z$~\cite{martin2021lie}. Thus, choosing a representation is equivalent to choosing an integer. This integer $q$ is called the charge of the particle fields of the associated bundle. The individual charge fields of matter $q$ are then seen to be sections of this associated bundle~\cite{bleecker2013gauge}.

The main goal of this work is to formalize the gauge theory of quantum thermodynamics~\cite{ThermoGauge1,ThermoGauge2} in the language of fiber bundles, thus writing thermodynamics on the same mathematical basis as the fundamental theories of physics.

\section{Geometric theory of quantum thermodynamics} \label{Geometrization}

We begin by addressing what might, at first, seem to be an inconsistency. The theory developed in~\cite{ThermoGauge1,ThermoGauge2} presents a gauge theory of the thermodynamic group $\mathrm{G}_T$. On the other hand, the potential $i \dot u_t u^\dagger_t$, defined for the covariant derivative, takes values in the Lie algebra of the unitary group $\mathrm{U}(d)$ and not in that of the thermodynamic group. Here, $\mathrm{U}(d)$ denotes the unitary group acting on the full system Hilbert space of dimension $d$, as defined in Sec. II. We will shortly see that this is not an inconsistency and that we are actually dealing with two distinct (although related) geometrical structures, and not just one.

Moreover, it is important to note that since $\mathrm{G}_T$ is not a fundamental gauge group, such as the ones in the standard model, but instead emerges from a lack of information about the system at hand, we should not expect it to behave in the same way as the group $\mathrm{U}(1)$ does in electrodynamics, for example. Indeed, as will become clear later on, this emergent thermodynamic gauge has a rich geometric structure that sets it apart from those commonly studied in field theories.

To construct the relevant fiber bundles for the theory, we note that all quantities of interest, such as the density operator $\rho_t$ and the Hamiltonian $H_t$, depend only on time. Therefore, by analogy with the gauge theories of the standard model, the base space of our bundle should be time, which can be taken to be the real line $\mathds{R}$. Since $\mathds{R}$ is contractible~\cite{nash2013topology}, our bundle is necessarily trivial by Theorem~\ref{contract} in Appendix~\ref{MathPrelim}. As such, we take the liberty to employ the isomorphism between a trivial bundle and a product bundle of the base space and the fiber, and we always regard quantum thermodynamics as taking place on the product bundle.

For the fiber, seeing the definition of the covariant derivative, we take the entire unitary group $\mathrm{U}(d)$ acting on the Hilbert space of dimension $d$. We further note that since $\mathrm{G}_T$ is a subgroup of $\mathrm{U}(d)$, any structure regarding the thermodynamic group can be considered a substructure of this principal $\mathrm{U}(d)$-bundle.

Thus, the first geometrical setup for quantum thermodynamics is a principal $\mathrm{U}(d)$ bundle $\xi=(\mathds{R}\times\mathrm{U}(d), \pi,\mathrm{U}(d),\mathds{R})$, with $\pi$ having the same meaning as before. We note that it makes sense for the gauge group to be $\mathrm{U}(d)$, as we are currently not concerned with any specific measurements of thermodynamic quantities that would require gauge invariance under the action of $\mathrm{G}_T$. Since this bundle arises from the unitary matrices $u_t$ that diagonalize the Hamiltonian $H_t$, we can observe that the gauge freedom associated with the unitary group is that of a choice of basis, which we use to work in the eigenbasis of the Hamiltonian.

The next step in characterizing the geometry of the bundle is verifying that the potential present in the covariant derivative defines a form of connection on $\xi$. This can be seen to be true by noting that if we slightly change the definition of the potential from $i \dot u_tu_t^\dagger$ to $A_t = \dot u_t u_t^\dagger$ (thus changing the covariant derivative to $\mathrm{D}_t \cdot = \partial_t + [A_t, \cdot]$), we see that $A_t$ is precisely a right-invariant Maurer-Cartan form on $\mathrm{U}(d)$. Thus, by Theorem~\ref{Connection} in Appendix~\ref{MathPrelim}, $A_t$ defines a connection on $\xi$. 

Thus, we have constructed a principal $\mathrm{U}(d)$-bundle $\xi$ and defined a connection on it. It is now time to understand how physical quantities such as $\rho_t$ and $H_t$ arise in this construction. To do this, we need to build an associated vector bundle.

We take $\mathds{V}$ to be the vector space of $d \times d$ Hermitian matrices and have the group $\mathrm{U}(d)$ act on it through the adjoint representation such that $u\cdot \nu = u\nu u^\dagger$, where $u \in \mathrm{U}(d)$ and $\nu \in \mathds{V}$. Since our original principal bundle is trivial, we can always take the associated bundle to be the product vector bundle with total space $\mathds{R} \times \mathds{V}$. This means that each point on $\eta$ is a tuple $(t, \nu)$, where $t$ is an instance of time and $\nu$ is some $d \times d$ Hermitian matrix. By analogy with the notion of matter fields as sections of the associated bundle discussed earlier, and noting that a global section of $\eta$ associates each point $t \in \mathds{R}$ with a Hermitian matrix $\nu \in \mathds{V}$, we see that the "matter fields" of the gauge theory of quantum thermodynamics are simply time-dependent Hermitian observables.

Given the above argument, we naturally consider the density matrices not as emergent gauge potentials, as proposed in~\cite{ThermoGauge1, ThermoGauge2}, but as specific matter fields of the theory. Thus, in this framework, the thermodynamic quantities of interest arise as functionals of the matter fields (observables), much like in a quantum field theory where several physical properties arise as functionals of matter fields (for example, Dirac or Klein-Gordon fields)~\cite{konopleva1981gauge,faddeev1991gauge,peskin1995introduction}.

With the construction we have so far, we can also see that, since the adjoint representation of a Lie group induces the adjoint representation of its Lie algebra, namely $A\cdot = [A, \cdot]$~\cite{kobayashi1996foundations}, via differentiation, the covariant derivative~\eqref{CovariantDerivative} implements that precise action, thus keeping itself consistent with the rest of the structure.

We now return to the topic of the gauge symmetry associated with the thermodynamic group $\mathrm{G}_T$. As defined in Eq.~\eqref{thermogroup}, it is more accurately described as a time-dependent family of Lie groups. Thus, for different points in the base space $\mathds{R}$, $\mathrm{G}_T$ is not necessarily the same group. As such, we do not, in general, have a single principal $\mathrm{G}_T$-bundle that can be defined as a sub-bundle of $\xi$ with base space $\mathds{R}$. However, we can note that at any instant $t$, we can form a trivial principal bundle with total space $\{t\} \times \mathrm{G}_T^t$, where by $\mathrm{G}_T^t$ we denote the thermodynamic group at time $t$. This bundle has a connection defined by the Maurer-Cartan form of $\mathrm{G}_T^t$, similarly to how we define the connection for the bundle $\mathrm{U}(d)$. It is this structure that carries along the emergent gauge symmetry related to the redundant information contained in the density matrix. This can be seen by noting that our construction of the associated bundle applies identically in this case.

This second geometric structure defined above formalizes the idea that the emergent gauge of quantum thermodynamics arises only concerning measurements, which we assume are taken at some specific time $t$. This can be understood as the geometric meaning of the coarse-graining introduced in Ref.~\cite{ThermoGauge1}. In a situation where the degeneracies of the Hamiltonian (and consequently the thermodynamic group) do not change continuously in time but vary in specific times $\{t_1, t_2, ..., t_m\}$, we can define, instead of bundles over points, bundles over time intervals $[0, t_1), [t_1, t_2), ... [t_{m-1}, t_m), [t_m, \infty)$ equipped with connections given by the Maurer-Cartan forms and Theorem~\ref{Connection}. This construction may be useful in a geometrical study of quench dynamics, where additional interaction terms in the Hamiltonian are turned on at a discrete set of times.

One important thing to note is that, since the curvature form is a 2-form and can also be written as the unique pullback of a 2-form in the base space~\cite{tu2017differential}, our principal bundles have globally vanishing curvature due to the fact that a 1-dimensional base space such as $\mathds{R}$ cannot have no non-vanishing 2-forms. 

The vanishing of the curvature should not be interpreted as a lack of physical content, but rather as a direct consequence of the operational structure of thermodynamics. In the present framework, the base manifold is identified with physical time, reflecting the fact that thermodynamic processes are defined through temporal protocols ---preparations, drivings, and measurements--- rather than through spatially extended gauge fields. As a result, the associated gauge structure organizes how thermodynamic descriptions are compared along time, rather than encoding dynamical forces in the sense of Yang–Mills theories.

Despite the absence of curvature, the thermodynamic connection plays a nontrivial physical role. It defines a covariant notion of temporal transport for thermodynamic quantities, ensuring consistency under time-dependent changes of description induced by restricted measurements. In particular, when the density matrix is interpreted as a section of an associated bundle, the connection specifies how states at different times are compared in a gauge-invariant manner. This covariant transport is essential for the definition of physically meaningful notions of work, heat, and entropy.

Importantly, the physical content of the theory resides in the global properties of the connection along finite-time processes rather than in local curvature. Even in one dimension, the connection, paired with the coarse-graining, can generate nontrivial holonomies that encode the cumulative effect of gauge transformations along a thermodynamic protocol. In this sense, irreversibility and path dependence arise not from curvature, but from the structure of time-ordered transport in the space of thermodynamic descriptions.

In an effort to make the discussion above easier to digest, we again invoke the electrodynamic example for a final comparison, as shown in Table~\ref{Comparison}.
\begin{table}[H]
    \centering
    \begin{tabular}{|c|c|}\hline
        \textbf{Electromagnetism}& \textbf{Quantum Thermodynamics} \\\hline
 $\mathrm{U}(1)$ symmetry&$\mathrm{U}(d)$ and $\mathrm{G}_T$ symmetries\\\hline
         4-vector potential& 
     Maurer-Cartan connection\\ \hline
 Electromagnetic field strength & Vanishing curvature\\\hline
 Charged particle fields & Hermitian observables\\ \hline\end{tabular}
    \caption{Comparison between the geometrical structure of the gauge theories of electrodynamics and quantum thermodynamics.}
    \label{Comparison}
\end{table}

The thermodynamic gauge group $\mathrm{G}_T$ is, in general, time dependent. This time dependence is not a formal artifact, but reflects a physically meaningful feature of thermodynamic processes: the information that is accessible to an observer may change during the evolution. In the present framework, the gauge group is determined by the degeneracy structure of the observables that are monitored during the process. Whenever this degeneracy structure changes in time, the associated thermodynamic gauge symmetry must also change.

Such situations naturally arise in driven quantum systems, where external control parameters modify the Hamiltonian and, consequently, the spectrum and its degeneracies. Examples include quenches, slow parameter ramps, or protocols crossing level crossings or symmetry-breaking points. In these cases, the equivalence classes of states defined by thermodynamic coarse-graining are themselves time dependent, and the appropriate description is not a single fixed gauge group acting globally in time, but a family of thermodynamic groups $\mathrm{G}_T$ acting locally along the process.

From a geometric perspective, this implies that the thermodynamic bundle is defined piecewise in time, with possibly distinct gauge structures over different temporal intervals. The connection introduced in this work provides a consistent way to relate thermodynamic descriptions across these intervals, ensuring gauge-covariant transport even when the underlying symmetry group changes. In this sense, the time dependence of $\mathrm{G}_T$ encodes changes in the informational constraints imposed by the measurement protocol, rather than changes in the microscopic dynamics itself.

This structure has direct physical implications. Since gauge invariance determines which quantities are thermodynamically meaningful, changes in $\mathrm{G}_T$ modify the set of admissible gauge-invariant quantities during the process. As a result, the classification of energy exchanges into work, heat, and entropy becomes explicitly protocol dependent, reflecting the evolving accessibility of information. The present framework thus provides a geometric language to describe nonequilibrium thermodynamic processes in which the notion of thermodynamic equivalence evolves in time.

\section{Application: The LMG model} 
\label{LMG}

To illustrate how the geometric thermodynamic framework developed in this work applies to a concrete physical system, we consider the Lipkin-Meshkov-Glick (LMG) model~\cite{Lipkin1965,Meshkov1965,Glick1965}. This model provides a particularly transparent example, as its symmetry properties and degeneracy structure are well understood and can be externally controlled through time-dependent parameters. Here we present a description of this model and relate it to the geometric structure constructed here. An extensive numerical analysis of this model was presented in the context of the gauge theory of thermodynamics in~\cite{ThermoGauge2}, where the invariant work, heat, and entropy were calculated for quench protocols.

The LMG model describes a system of $N$ spin-$\tfrac{1}{2}$ particles with infinite-range interactions and is conveniently expressed in terms of collective spin operators $J_\alpha = \tfrac{1}{2}\sum_{i=1}^N \sigma_i^\alpha$, where $\alpha = x,y,z$. A commonly used form of the Hamiltonian reads
\begin{equation}
H_t = -\frac{\lambda_t}{N}\left(J_x^2 + \gamma J_y^2\right) - h_t\, J_z ,
\end{equation}
where $\lambda_t$ is an interaction strength, $\gamma$ is an anisotropy parameter, and $h_t$ is a transverse field. Both $\lambda_t$ and $h_t$ can be controlled externally and vary over time according to a prescribed protocol.

For fixed parameters, the Hamiltonian is invariant under permutations of spins and conserves the total angular momentum. As a consequence, the Hilbert space decomposes into invariant subspaces labeled by the total spin quantum number $J$, and the energy spectrum generally exhibits degeneracies associated with this symmetry. These degeneracies define the equivalence classes of states under thermodynamic coarse graining and, in the present framework, determine the thermodynamic gauge group $\mathrm{G}_T$ acting on the associated bundle of density operators.

When the system is driven by a time-dependent protocol, the degeneracy structure of the Hamiltonian may change. For example, variations in the transverse field $h_t$ or the interaction strength $\lambda_t$ can lift or create degeneracies, particularly near critical points that separate distinct phases of the model. As a result, the thermodynamic gauge group becomes explicitly time dependent, $\mathrm{G}_T \rightarrow \mathrm{G}_T(t)$, reflecting the evolving symmetry and information content of the system.

Within the geometric formulation developed in this work, this situation is naturally described by a family of thermodynamic gauge groups defined along the time axis. The density matrix is interpreted as a section of an associated bundle, and the thermodynamic connection ensures gauge-covariant transport of states under time evolution, even when the underlying gauge group changes. Importantly, although the base manifold is one-dimensional and the curvature vanishes identically, the connection remains nontrivial and encodes the physically relevant parallel transport between thermodynamic descriptions defined at different times.

The physical implications of this structure are directly observable at the level of thermodynamic quantities. Since gauge invariance determines which observables are thermodynamically significant, changes in $\mathrm{G}_T$ modify the set of admissible gauge-invariant quantities during the protocol. In particular, the decomposition of energy variations into work and heat becomes explicitly dependent on the evolving thermodynamic symmetry, reflecting changes in the informational constraints imposed by the measurement scheme. This provides a geometric explanation for the protocol dependence of thermodynamic quantities in driven many-body systems.

The LMG model thus offers a concrete realization of the central ideas of this work. It shows how changes in degeneracy and symmetry naturally induce a time-dependent thermodynamic gauge structure and how the associated bundle formulation provides a consistent and physically transparent description of nonequilibrium thermodynamic processes. Although no explicit dynamics is required for the present discussion, this example demonstrates that the refined geometric framework developed here extends beyond a purely formal construction and captures physically relevant features of realistic many-body systems.

The time dependence of the thermodynamic gauge group $\mathrm{G}_T(t)$ has direct consequences for the definition and interpretation of thermodynamic quantities such as work, heat, and entropy. In the present framework, these quantities are defined through gauge-covariant variations of the density matrix along the time axis. As a result, their physical meaning is intrinsically tied to the connection on the thermodynamic bundle.

For a fixed thermodynamic gauge group, the infinitesimal change of the internal energy,
\begin{equation}
\dd U = \Tr\!\left(\rho\, \dd H\right) + \Tr\!\left(H\, \dd \rho\right),
\end{equation}
admits a gauge-invariant decomposition into work and heat, with the first term identified as work and the second as heat. This separation relies on the assumption that the notion of thermodynamic equivalence, encoded in $\mathrm{G}_T$, remains unchanged during the process.

When the thermodynamic gauge group becomes time-dependent, this assumption no longer holds globally. Changes in the degeneracy structure of the Hamiltonian modify the equivalence classes of states under thermodynamic coarse-graining and therefore alter which variations of the density matrix are physically distinguishable. In geometric terms, the subspaces defined by the connection are themselves time dependent.

As a consequence, the gauge-covariant derivative of the density matrix acquires additional contributions associated with the change of the gauge structure. These contributions reflect the fact that part of the state variation arises not from physical energy exchange with an environment but from a change in the informational constraints defining the thermodynamic description. In this situation, the standard identification of $\Tr(H\,\dd \rho)$ as heat must be refined to include the geometric contribution induced by the evolving connection.

This effect has a clear physical interpretation in the context of the LMG model. When a time-dependent protocol lifts or creates degeneracies in the energy spectrum, previously indistinguishable microstates become distinguishable, or vise versa. This change modifies the entropy associated with coarse-grained descriptions, even in the absence of energy exchange~\cite{ThermoGauge2}. The resulting entropy variation is therefore geometric in origin, arising from parallel transport between thermodynamic descriptions defined by different gauge groups.

From this perspective, entropy production receives a contribution associated with the mismatch between thermodynamic descriptions at different times, encoded in the holonomy generated by the thermodynamic connection and coarse-graining. Although the curvature vanishes identically due to the one-dimensional base manifold, the connection still generates nontrivial parallel transport, which captures the cumulative effect of changing informational constraints along the process.

The LMG model thus illustrates how the refined geometric structure developed in this work provides a unified description of energetic and informational contributions to thermodynamic quantities. Work and heat are no longer determined solely by microscopic dynamics, but also by the evolution of the thermodynamic gauge structure. This demonstrates that the connection introduced here is not merely a mathematical artifact but a physically meaningful object encoding how thermodynamic quantities transform under time-dependent coarse-graining and symmetry changes.

\section{Discussion} 
\label{Discussion}

In this work, we have mathematically formalized the gauge theory of quantum thermodynamics, whose main ideas were put forward in Refs.~\cite{ThermoGauge1,ThermoGauge2}. Specifically, the theory was formulated through the geometry of a principal $\mathrm{U}(d)$-bundle with an additional geometric structure associated with measurements (namely the thermodynamic group $\mathrm{G}_T$).

In our argument, we have separated which mathematical objects belong to each geometric structure associated with the theory. Furthermore, by constructing associated vector bundles, we were able to create an analogy between gauge-theoretic matter fields and time-dependent Hermitian observables in quantum thermodynamics. This result helps to strengthen the link between the theory proposed in Refs.~\cite{ThermoGauge1,ThermoGauge2} and the gauge theories in fundamental physics, such as Yang-Mills theory. In other words, we put thermodynamics in the same mathematical language as fundamental theories in physics.

Furthermore, although the structure associated with the thermodynamic group does not, in general, form a single principal fiber bundle, it may form a collection of them if the degeneracies of the observables under consideration change discretely over time. In particular, formalizing these ideas geometrically allows us to propose a series of questions regarding how the topological properties of the thermodynamic group influence the thermodynamics of a quantum system.

Finally, we note that thermodynamic quantities depend only on energy differences and are therefore invariant under global shifts of the Hamiltonian by a multiple of the identity. Within the present framework, such shifts correspond to transformations that leave all gauge-invariant thermodynamic observables unchanged. This reflects the fact that the thermodynamic gauge structure encodes redundancies associated with informational redundant degrees of freedom in the description, including both coherences within degenerate subspaces and global energy offsets. The proposed gauge formulation thus naturally incorporates the fundamental insensitivity of thermodynamics to the choice of energy zero.

One natural example of a question that arises from this theory is whether the homology and cohomology of the thermodynamic group can lead to any insight into the possible dynamics of a system. Since the thermodynamic group is a product of Lie groups, it should be possible to study this issue by employing Kunneth's theorem~\cite{nash2013topology}.

Another perhaps less obvious direction is the study of relative homology and cohomology~\cite{nash2013topology}. Since the thermodynamic group is a Lie subgroup of the unitary group $\mathrm{U}(d)$, relative homologies and cohomologies between them can be constructed. This may be of interest, as both $\mathrm{G}_T$ and $\mathrm{U}(d)$ are relevant gauge groups of the theory. In particular, since relative Rham cohomology relates to differential forms in $\mathrm{U}(d)$ that vanish in $\mathrm{G}_T$, and gauge-invariant quantities in this framework carry integrals over the thermodynamic group, it makes sense that it may have some meaning in theory.

Additionally, we know that in the theory of classical thermodynamics, there are many important inequalities, such as inequalities that restrict changes in the free energies of a system~\cite{Callen1991}. Meanwhile, in topology, there are known relations between the critical points of functions and the topology of their domain. That is, Morse theory~\cite{nash2013topology} does so by constructing inequalities that relate the numbers of critical points with different indices of the function to the Betti numbers (ranks of homology and cohomology groups) of the space. From this perspective, we believe that Morse theory may offer insight into quantum thermodynamic equilibrium by examining the critical points of certain properly defined free energies. In this context, it may be possible to derive a generalized Clausius inequality, such as the one hypothesized in~\cite{ThermoGauge2}, that is related to the structure of the thermodynamic group through its algebraic or geometric/topological properties~\cite{Pernambuco2026}. 

The last, perhaps less speculative direction to follow relates to the study of the thermodynamic group in the context of Lie group bouquets~\cite{douady1966espaces}. As discussed earlier, due to its time dependence, the thermodynamic group by itself does not form a proper fiber bundle for all times (and thus does not define a gauge theory). The concept of a Lie Group Bouquet, introduced in the seminal 1966 paper by Douady and Lazard~\cite{douady1966espaces}, formalizes the notion of a family of Lie groups parametrized by some other manifold (which, in this case, would be $\mathds R$). It is immediately obvious that the thermodynamic group forms one such bouquet, and it should be worthwhile to study it as such.

\section{Acknowledgments}

The authors are indebted to Sebastian Deffner, Rômulo César Rougemont Pereira, Rafael Chaves Souto Araujo, and Thiago Rodrigues Oliveira for critically reading and commenting on the manuscript. LCC acknowledges CNPq through grant 308065/2022-0, the financial support of the National Institute of Science and Technology for Applied Quantum Computing through CNPq grant 408884/2024-0, and the warm hospitality of the International Institute of Physics, where this work was developed with the support of the Simons Foundation (Grant No. 1023171, R.C.).

\appendix

\section{Mathematical background} \label{MathPrelim}

We briefly review the notion of Lie groups and fiber bundles. We focus on the topics we have employed in the main part of the text. The goal is to make the article more self-contained and accessible to researchers interested in quantum thermodynamics, information theory, and gauge theories. 

The theory of Lie groups is discussed in Refs.~\cite{martin2021lie,tu2017differential}, while nice references on fiber bundles are~\cite{Baez1994,husemoller2013fiber,nash2013topology}. See also \cite{bleecker2013gauge,konopleva1981gauge} for some topics related to principal bundles. We direct the interested reader to these references for a deeper treatment.

\subsection{Bundles and fiber bundles}

We provide a brief description of the main mathematical structure that will be employed to build a geometric theory of quantum thermodynamics. 

A bundle is a mathematical structure that generalizes the notion of a product space. More specifically, a bundle can be understood as a triple $(\mathds E, \pi, \mathds B)$, with $\mathds E$ being a topological space called the total space, $\mathds{B}$ a topological space called the base space, and $\pi:\mathds E\rightarrow\mathds B$ the projection map.

For every point $b \in\mathds B$, we call the space $\pi^{-1}(b)$ the fiber of the bundle over $b$. It is often useful to think of a bundle as a set of spaces (fibers) parametrized by the base space and glued together by the topology of the total space.

Bundles are the basic mathematical objects over which many other, more complex structures, such as vector bundles, fiber bundles, and principal bundles, are built. In this work, we are mainly interested in principal bundles. However, it is often pedagogical to build the notion of a principal bundle from that of a fiber bundle. Thus, we define a fiber bundle by adding two elements to the bundle. A topological space $\mathds F$, called the fiber, and a group $\mathrm G$ of homeomorphisms of the fiber $\mathds F$, called the structure group. Moreover, there is an open cover $\{O_{\alpha}\}$ of $\mathds B$ such that, at each $O_{\alpha}$, $\pi^{-1}(O_{\alpha})$ is homeomorphic to the Cartesian product $O_{\alpha}\times \mathds F$. The homeomorphisms $\phi_\alpha:\pi^{-1}(O_\alpha) \rightarrow O_\alpha\times \mathds F$ are called local trivializations of $\mathds E$. 

Thus, a fiber bundle can be seen as a topological space that is locally (but generally not globally) a product space of $\mathds B$ and $\mathds F$. If a fiber bundle $\xi$ is globally a product space (i.e, if $\mathds E = \mathds B \times \mathds F$), then $\xi = (\mathds B \times \mathds F, \pi, \mathds F, \mathrm G, \mathds B)$ is called a product bundle. A fiber bundle $\zeta = (\mathds E', \pi', \mathds F, \mathrm G, \mathds B)$ is called trivial if it is isomorphic to the product bundle $\xi$.

A section of a fiber bundle is a map $s:\mathds B \rightarrow \mathds E$ that inverts the projection onto the base space. More specifically, we have $p\circ s = \mathds{1}_{\mathds{B}}$, where $\mathds{1}_{\mathds{B}}$ is the identity map on $\mathds B$. The condition that a bundle $\xi$ is trivial is equivalent to the existence of a global section on $\xi$.

We now state an important theorem:
\begin{theorem}
Every fiber bundle on a contractible base space $\mathds{B}$ is trivial.
\label{contract}
\end{theorem}
By contractible, we mean that the identity map in $\mathds{B}$ is homotopic to some constant map. In other words, the topological space can be continuously deformed into a single point.

We do not prove this theorem here, since it would require many concepts that are outside the scope of this work. For those interested in the proof, Ref.~\cite{nash2013topology} provides an intuitive outline of why Theorem~\ref{contract} is true, while Ref.~\cite{steenrod1999topology} provides a rigorous proof using homotopy theory.

A crucial example of a (generally non-trivial) fiber bundle associated with a manifold $\mathds{M}$ is the tangent bundle $\mathrm{T}\mathds{M}$, defined as
\[
\mathrm{T}\mathds{M} = \bigcup\limits_{p \in \mathds{M}}\mathrm{T}_p\mathds{M},
\] 
where $\mathrm{T}_p\mathds{M}$ is the tangent space of $\mathds{M}$ at the point $p\in\mathds{M}$. A manifold $\mathrm{M}$ is called parallelizable if its tangent bundle is trivial~\cite{martin2021lie}.

\subsection{Lie Theory}

A Lie group $\mathrm{G}$ is a group that is also a differentiable manifold in such a way that the map 
\begin{equation}
    p:(g, h) \in \mathrm{G} \times \mathrm{G} \rightarrow gh \in \mathrm{G}
\end{equation}
is smooth. In particular, as a consequence of the condition above, the operation of taking an inverse in a Lie group is also differentiable (more precisely, it is a diffeomorphism). Since both the operations of taking a product and an inverse are differentiable, they are continuous; thus, all Lie groups are topological groups.

The Lie algebra $\mathfrak{g}$ associated with a Lie group $\mathrm{G}$ is the space of vector fields invariant under the action of $\mathrm{G}$, with a bracket given by the Lie bracket of vector fields. However, for our purposes, it is more convenient to exploit the fact that the Lie algebra of a Lie group is isomorphic to $\mathrm{T}_e\mathrm{G}$, the tangent space of the group at the identity element $e$, and to define the Lie algebra directly as such. Although the Lie algebra is isomorphic to the tangent space at the identity, there is a more general relation between Lie algebras and the tangent spaces of their associated Lie groups, namely:
\begin{equation}
    \mathrm{T}_g\mathrm{G} \cong \mathrm{G} \times \mathfrak{g}.
\end{equation}
Since $\mathfrak{g} \cong \mathrm{T}_e\mathrm{G}$, this also means that every Lie group is parallelizable.

The triviality of the tangent bundle of a Lie group $\mathrm{G}$ allows us to define the so-called left- and right-invariant Maurer-Cartan forms:
\begin{equation}
    \theta^L_g(v) = \dd(L_{g^{-1}})_g(v)
\end{equation}
and
\begin{equation}
    \theta^R_g(v) = \dd(R_{g^{-1}})_g(v),
\end{equation}
respectively. $\dd$ represents differentiation, $L_g$ and $R_g$ are, respectively, left- and right-translated by $g \in \mathrm{G}$ and $v \in \mathrm{T}_g \mathrm{G}$. For a matrix Lie group, such as the ones we will be interested in for this work, the Maurer-Cartan forms can be expressed extrinsically as follows:
\begin{equation}
    \theta_g^L = g^{-1}\cdot \dd g;
\end{equation}
and
\begin{equation}
    \theta_g^R = \dd g\cdot g^{-1}.
    \label{Maurer-Cartan}
\end{equation}
Maurer-Cartan forms play an important role in the gauge theory of the thermodynamic group presented in the main part of the paper.

Given a Lie group $\mathrm{G}$ and a manifold $\mathrm{M}$, a smooth right action of $\mathrm{G}$ on $\mathrm{M}$ is a smooth map $\mu\, : \, \mathds{M}\times \mathrm{G}\rightarrow \mathds{M}$, which we will denote $x\cdot g = \mu(x, g)$, such that:
\begin{enumerate}
    \item $x \cdot e = x$;
    \item $(x \cdot g) \cdot h = x \cdot (gh)$
\end{enumerate}
$\forall x \in \mathds{M}$, $g, h \in \mathrm{G}$. $e$ is the identity element in $\mathrm{G}$. We say that an action is free if $x\cdot g = x \implies g = e$.

Given two manifolds $\mathds{M}$ and $\mathds{N}$ where $\mathrm{G}$ acts on the right, a map $f\,:\,\mathds{M} \rightarrow \mathds{N}$ is said to be $\mathrm{G}$-equivariant if $f(x\cdot g) = f(x) \cdot g$, where $(x, g) \in \mathds{M} \times \mathrm{G}$.

\subsection{Principal Bundles}

Principal bundles are among the most important bundles that arise in physics, as they provide the mathematical structure on which one can formalize gauge theories as geometrical constructs.

A principal $\mathrm{G}$-bundle is a fiber bundle for which the structure group (a Lie group $\mathrm{G}$) and the fiber are one and the same. More specifically, given a Lie group $\mathrm{G}$, we define a principal $\mathrm{G}$-bundle $(\mathds{E}, \pi, \mathrm{G}, \mathds{B})$ as a smooth fiber bundle $(\mathds{E}, \pi, \mathds{F}, \mathrm{G}, \mathds{B})$ such that \cite{tu2017differential}:
\begin{enumerate}
    \item $\mathds{F} = \mathrm{G}$;
    \item $\mathrm{G}$ acts freely on the right in $\mathds{E}$;
    \item Local trivializations $\phi_\alpha:\pi^{-1}(O_\alpha) \rightarrow O_\alpha \times \mathrm{G}$ are $\mathrm{G}$-equivariant, where $\mathrm{G}$ acts on $O_\alpha \times \mathrm{G}$ by $(b, h)\cdot g = (b, hg)$.
\end{enumerate}

The most important objects in the study of the geometry of a principal $\mathrm{G}$-bundle are connections. More specifically, given a principal $\mathrm{G}$-bundle $\xi = (\mathds{E}, \pi, \mathrm{G}, \mathds{B})$, we say that a $\mathfrak g$-valued $1$-form $\omega$ in $\mathds{E}$ is a connection if and only if 
\begin{enumerate}
    \item $\omega$ is $\mathrm{G}$-equivariant;
    \item $\omega(A_p^*) = A$, where 
    \[
    A_p^* = \frac{\dd}{\dd t}(p\cdot e^{tA})|_{t = 0},
    \]
    for $A \in \mathfrak g$ and $p \in \mathds{E}$.
\end{enumerate}

The $A_p^*$ in the above definition are vector fields on $\mathds{E}$ called fundamental vector fields. Connections are very important in differential geometry because they allow us to define the notion of parallel transport in general spaces.

Given a connection $\omega$ on a principal $\mathrm{G}$-bundle, one can always construct a curvature form associated with it. That is, from $\omega$, one can make a $\mathfrak g$-valued 2-form that acts on the vector fields $X$ and $Y$ over the total space $\mathds{E}$ in the following manner:
\begin{equation}
    \Omega(X, Y) = \dd\omega(X, Y) + [\omega(X), \omega(Y)],
    \label{CurvatureForm}
\end{equation}
where $[\cdot, \cdot]$ is the Lie bracket of $\mathfrak g$.

The notion of a connection on a principal $\mathrm{G}$-bundle makes evident the significance of the Maurer-Cartan forms for the geometry of a Lie group. That is, if we consider a Lie group $\mathrm{G}$ as the principal $\mathrm{G}$-bundle product with total space $\{x\}\times \mathrm{G}$, where $\{x\}$ is a single point, then $\theta_g^R$ is a connection on the Lie group regarded as a bundle.

We now turn our attention to a result on connections on trivial principal $\mathrm{G}$-bundles that can be found in~\cite{tu2017differential,tu2020introductory} and is useful in our construction.
\begin{theorem}
    Let $\xi = (\mathds{B} \times \mathrm{G}, \pi_1, \mathrm{G}, \mathds{B})$ be a trivial principal $\mathrm{G}$-bundle, and $\theta_g^R$ be the right-invariant Maurer-Cartan form on $\mathrm{G}$. Define projection $\pi_2:\mathds{B}\times\mathrm{G} \rightarrow\mathrm{G}$, $\pi_2(x, g) = g$. Then the pullback $ \omega_g = \pi_2^*(\theta_g^R)$ is a connection on $\xi$.
\label{Connection}
\end{theorem}

\begin{proof}
First, we prove that $\omega_g$ is $\mathrm{G}$-equivariant. Let $(X_x, X_g)$ be a vector in the tangent space $\mathrm{T}_{(x, g)}\mathds{B}\times \mathrm{G} \cong \mathrm{T}_x\mathds{B} \oplus \mathrm{T}_g\mathrm{G}$ and $h \in \mathrm{G}$. Then $\omega_g(X_x, X_g) \cdot h = \pi_2^*\theta_g^R(X_x, X_g) \cdot h = \theta_g^R(\pi_2(X_x, X_g)) \cdot h = \theta_g^R(X_g)\cdot h$. But since $\theta_g^R$ is a connection on $\mathrm{G}$ regarded as a bundle over a single point, $\theta_g^R(X_g)\cdot h = \theta_g^R(X_g \cdot h) = \omega_g((X_x, X_g) \cdot h)$. Thus, $\omega_g$ is $\mathrm{G}$-equivariant.

Now, let $A_p^* = \frac{\dd}{\dd t}((x, g)\cdot e^{tA})|_{t=0}$ be a fundamental vector field on $\mathds{B} \times \mathrm{G}$. Then
\begin{eqnarray*}
\omega_g(A_p^*) &=& \pi_2^*\theta_g^R(\frac{\dd}{\dd t}((x, g)\cdot e^{tA})|_{t=0}) \\
&=& \theta_g^R(\pi_2(\frac{\dd}{\dd t}((x, g)\cdot e^{tA})|_{t=0}) \\
&=& \theta_g^R(\frac{\dd}{\dd t}(g\cdot e^{tA})|_{t=0}).
\end{eqnarray*}

But this is exactly the Maurer-Cartan form applied to a fundamental vector field of $\mathrm{G}$. Thus
\[
\omega_g(A_p^*) = \theta_g^R(\frac{d}{d t}(g\cdot e^{tA})|_{t=0}) = A.
\]

Since $\omega_g$ satisfies both axioms of a connection on a principal $\mathrm{G}$-bundle, it is a connection on $\xi$. 
\end{proof}


\begin{thebibliography}{10}

\bibitem{Callen1991} H. B. Callen, {\em Thermodynamics and an Introduction to Thermostatistics} (Wiley, 1991).

\bibitem{Strasberg2022} P. Strasberg, {\em Quantum Stochastic Thermodynamics: Foundations and Selected Applications} (Oxford University Press, 2022).

\bibitem{Goold_2016} J. Goold, M. Huber, A. Riera, L. del Rio, and P. Skrzypczyk. The role of quantum information in thermodynamics --- a topical review. J. Phys. A: Math. Theor., \textbf{49}, 143001 (2016).

\bibitem{Alicki1979} R. Alicki. The quantum open system as a model of the heat engine. J. Phys. A: Math. Gen. \textbf{12}, L103 (1979).

\bibitem{Alicki2018} R. Alicki and R. Kosloff, Introduction to quantum thermodynamics: History and prospects, in F. Binder, L. A. Correa, C. Gogolin, J. Anders, and G. Adesso (editors), {\em Thermodynamics in the Quantum Regime: Fundamental Aspects and New Directions}, volume 195 of {\em Fundamental Theories of Physics} (Springer International Publishing, 2018).

\bibitem{Deffner2019} S. Deffner and S. Campbell, {\em Quantum Thermodynamics: An Introduction to the Thermodynamics of Quantum Information} (IOP Concise Physics, 2019).

\bibitem{Baez1994} J. Baez and J. P. Muniain, {\em Gauge Fields, Knots and Gravity: 4} (World Scientific Publishing Company, 1994).

\bibitem{konopleva1981gauge} N. P. Konopleva, V. N. Popov, and N. M. Queen, {\em Gauge Fields} (Harwood Academic Publishers, 1981).

\bibitem{faddeev1991gauge} L. D. Faddeev and A. A. Slavnov, {\em Gauge Fields: An Introduction To Quantum Theory}, (Frontiers in physics. Basic Books, 1991).

\bibitem{nash2013topology} C. Nash and S. Sen, {\em Topology and Geometry for Physicists}, (Dover Books on Mathematics, Dover Publications, 2013).

\bibitem{bleecker2013gauge} D. Bleecker, {\em Gauge Theory and Variational Principles}, (Dover Books on Mathematics, Dover Publications, Incorporated, 2013).

\bibitem{PhysRevLett.58.2051} A. Shapere and F. Wilczek. Self-propulsion at low reynolds number. Phys. Rev. Lett. \textbf{58}, 2051 (1987).

\bibitem{Shapere_Wilczek_1989} A. Shapere and F. Wilczek. Geometry of self-propulsion at low reynolds number. J. Fluid Mech. \textbf{198}, 557 (1989).

\bibitem{Fradkin2013} E. Fradkin, {\em Field Theories of Condensed Matter Physics} (Cambridge University Press, 2013).

\bibitem{GNNGauge2021} D. Luo, G. Carleo, B. K. Clark, and J. Stokes. Gauge equivariant neural networks for quantum lattice gauge theories. Phys. Rev. Lett. \textbf{127}, 276402 (2021).

\bibitem{Vazquez2012} S. E. V\'{a}zquez and S. Farinelli. Gauge invariance, geometry and arbitrage. J. Invest. Strateg. \textbf{1}, 23 (2012).

\bibitem{Katagiri2018} S. Katagiri. Non-equilibrium thermodynamics as gauge fixing. Prog. Theor. Exp. Phys. \textbf{2018}, 093A02 (2018).

\bibitem{Borlenghi2016} S. Bolenghi. Gauge invariance and geometric phase in nonequilibrium thermodynamics. Phys. Rev. E \textbf{93}, 012133 (2016).

\bibitem{polettini2012nonequilibrium} M. Polettini. Nonequilibrium thermodynamics as a gauge theory. EPL \textbf{97}, 30003 (2012).

\bibitem{ThermoGauge1} L. C. C\'eleri and \L . Rudnicki. Gauge-invariant quantum thermodynamics: Consequences for the first law.
Entropy, \textbf{26} (2024).

\bibitem{ThermoGauge2} G. F. Ferrari, \L . Rudnicki, and L. C. C\'eleri. Quantum thermodynamics as a gauge theory. Phys. Rev. A \textbf{111}, 052209 (2025).

\bibitem{Araujo2018} R. M. Ara\'{u}jo, T. H\"{a}ffner, R. Bernardi, D. S. Tasca, M. P. J. Lavery, M. J. Padgett, A. Kanaan, L. C. C\'{e}leri, and P. H. Souto~Ribeiro. Experimental study of quantum thermodynamics using optical vortices. J. Phys. Commun. \textbf{2}, 035012 (2018).

\bibitem{tu2017differential} L. W. Tu, {\em Differential Geometry: Connections, Curvature, and Characteristic Classes} (Graduate Texts in Mathematics, Springer International Publishing, 2017).

\bibitem{martin2021lie} L. A. B. S. Martin, {\em Lie Groups} (Latin American Mathematics Series, Springer International Publishing, 2021).

\bibitem{peskin1995introduction} M. E. Peskin and D. V. Schroeder, {\em An Introduction To Quantum Field Theory} (Frontiers in Physics, Avalon Publishing, 1995).

\bibitem{kobayashi1996foundations} S. Kobayashi and K. Nomizu, {\em Foundations of Differential Geometry, Volume 1} (Wiley Classics Library, Wiley, 1996).

\bibitem{Lipkin1965} H. Lipkin, N. Meshkov, and A. Glick, Validity of manybody approximation methods for a solvable model: (i). exact solutions and perturbation theory, Nuclear Physics \textbf{62}, 188 (1965).

\bibitem{Meshkov1965} N. Meshkov, A. Glick, and H. Lipkin, Validity of manybody approximation methods for a solvable model: (ii). linearization procedures, Nuclear Physics \textbf{62}, 199 (1965).

\bibitem{Glick1965} A. Glick, H. Lipkin, and N. Meshkov, Validity of many body approximation methods for a solvable model: (iii). diagram summations, Nuclear Physics \textbf{62}, 211 (1965).

\bibitem{douady1966espaces} A. Douady and M. Lazard, Espaces fibr{\'e}s en algebres de lie et en groupes. Inven. Math. \textbf{1}, 133 (1966).

\bibitem{husemoller2013fiber} D. Husem{\"o}ller, {\em fiber Bundles} (Graduate Texts in Mathematics, Springer New York, 2013).

\bibitem{steenrod1999topology} N. Steenrod, {\em The Topology of fiber Bundles} (Princeton Landmarks in Mathematics and Physics, Princeton University Press, 1999).

\bibitem{tu2020introductory} L. W. Tu, {\em Introductory Lectures on Equivariant Cohomology} (Annals of Mathematics Studies, Princeton University Press, 2020).

\bibitem{Pernambuco2026} T. Pernambuco and L. C. Céleri, Geometry of restricted information: the case of quantum thermodynamics, https://www.arxiv.org/abs/2602.06716 (2026).

\end{thebibliography}
\end{document}